\documentclass[12pt]{amsart}
\usepackage[utf8]{inputenc}
\usepackage{amssymb, amsmath}
\usepackage{fourier}
\usepackage{enumitem}

\usepackage{geometry}
\theoremstyle{plain}

\newtheorem{theorem}{Theorem}

\newtheorem{claim}[theorem]{Claim}

\newtheorem{definition}{Definition}

\theoremstyle{remark}

\newtheorem{remark}{Remark}

\newcommand{\bit}{\begin{itemize}}
\newcommand{\eit}{\end{itemize}}
\newcommand{\ben}{\begin{enumerate}}
\newcommand{\een}{\end{enumerate}}
\newcommand{\be}{\begin{equation}}
\newcommand{\ee}{\end{equation}}
\newcommand{\ba}{\begin{array}}
\newcommand{\ea}{\end{array}}

\newcommand{\mc}[1]{\mathcal{#1}}

\newcommand{\dt}{\mathrm{d}t}
\newcommand{\ds}{\mathrm{d}s}

\newcommand{\eps}{\varepsilon}

\newcommand{\inner}[2]{\left\langle #1,#2 \right\rangle}

\newcommand\cB{\mathcal B}
\newcommand\cH{\mathcal H}

\newcommand\C{\mathbb C}
\newcommand\R{\mathbb R}
\newcommand\N{\mathbb N}

\renewcommand\l{\lambda}
\newcommand\bh{\cB(\cH)}

\newcommand\tr{\operatorname{Tr}}
\newcommand{\ler}[1]{\left( #1 \right)}

\begin{document}

\title{Jointly convex quantum Jensen divergences}

\author{D\'aniel Virosztek}
\address{Institute of Science and Technology Austria\\
Am Campus 1, 3400 Klosterneuburg, Austria}
\email{daniel.virosztek@ist.ac.at}
\urladdr{http://pub.ist.ac.at/\~{}dviroszt}

\thanks{The author was supported by the ISTFELLOW program of the Institute of Science and Technology Austria (project code IC1027FELL01) and partially supported by the Hungarian National Research, Development and Innovation Office – NKFIH (grant no. K124152).}

\dedicatory{To my grandparents, Ir\'en Engler and Dr. Ferenc Erd\H{o}si}


\begin{abstract}
We investigate the quantum Jensen divergences from the viewpoint of joint convexity. It turns out that the set of the functions which generate jointly convex quantum Jensen divergences on positive matrices coincides with the Matrix Entropy Class which has been introduced by Chen and Tropp quite recently \cite{chen-tropp}.
\end{abstract}
\maketitle

\section{Introduction} \label{sec:intro}

\subsection{Motivation}
Measuring the dissimilarities of objects is of a particular importance in many areas of mathematics and mathematical physics. From our viewpoint, the most important area with the above property is the field of classical and quantum information theory.
\par
The measure of the dissimilarity of elements of a certain space may be a true metric, but there are several important examples when this is not the case. In classical information theory, the most fundamental quantity measuring the dissimilarity of probability distributions is the \emph{relative entropy} (or \emph{Kullback-Leibler divergence} \cite{kullback-leibler}) which is far from being a true metric. Indeed, the relative entropy is not symmetric and it does not satisfy the triangle inequality either.
Despite these "defects", the Kullback-Leibler divergence (and its non-commutative analogue, the \emph{Umegaki relative entropy}) has several very useful properties which make it a central quantity of the classical (respectively, quantum) information theory. One of these useful properties is the \emph{joint convexity}. (A two-variable function $h$ defined on a convex set $C$ is said to be jointly convex if $h\ler{\sum_{j=1}^m \alpha_j x_j, \sum_{j=1}^m \alpha_j y_j} \leq  \sum_{j=1}^m \alpha_j h \ler{x_j, y_j}$ holds for any $x_1, \dots, x_m, y_1, \dots, y_m \in C$ and $\alpha_1, \dots, \alpha_m \in [0,1]$ with $\sum_{j=1}^m \alpha_j =1.$) For \emph{homogeneous} relative entropy-type quantities, joint convexity is equivalent to the monotonicity under stochastic maps (both in classical and in quantum information theory; the non-commutative counterpart of the stochastic map is the \emph{completely positive trace preserving} (CPTP) map). For details, see \cite[remarks after Def. 2.3]{les-rus}.
\par
The aim of this paper is to investigate \emph{quantum Jensen divergences} from the viewpoint of joint convexity. Given a convex function $\varphi$ on a convex set $C$ and a number $\lambda \in (0,1),$ the Jensen divergence (which will be denoted by $j_{\varphi, \l}(.,.)$) of the objects $a,b \in C$ is nothing else but the difference between the two sides of the Jensen inequality, that is,
$$
j_{\varphi, \l}(a,b)=(1-\l)\varphi(a) + \l \varphi(b) -\varphi\ler{(1-\l)a+\l b}.
$$
By quantum Jensen divergences we mean Jensen divergences of self-adjoint matrices which are generated by functions of the form $A \mapsto \tr f(A).$ (Note that the convexity of the function $f: \R \rightarrow \R$ ensures the convexity of the map $A \mapsto \tr f(A)$ on the set of self-adjoint matrices, see, e.g., \cite[Thm. 2.10.]{carlen}.) For the precise definition of quantum Jensen divergences, see Definition \ref{def:jensen-div}.
\par
In quantum information theory, we are mainly interested in Jensen divergences of positive matrices (or even more restictively, density matrices). Certain special quantum Jensen divergences (e.g., the quantum Jensen-Shannon divergence) have interesting metric properties, as well (see \cite{briet-harremoes} and \cite{lamberti}). The generalized isometries of the Jensen divergences on positive definite matrices have been determined in \cite{mpv-15}. 
\par
The main result of our paper is that for a convex function $f$ defined on the positive half-axis the quantum Jensen divergence induced by the map $A \mapsto \tr f(A)$ is jointly convex on the set of positive definite matrices if and only if the generating function $f$ belongs to the \emph{Matrix Entropy Class} which is an interesting function class introduced by Chen and Tropp in \cite{chen-tropp} quite recently. This result may be considered as a continuation of our previous work on the joint convexity of \emph{Bregman divergences} \cite{pv15}.

\subsection{Basic notions, notation}
We denote the set of all complex $n \times n$ matrices by $M_n(\C)$ and the sets of all self-adjoint, positive semidefinite, and positive definite complex $n \times n$ matrices are denoted by $M_n^{sa}(\C), \, M_n^{+}(\C),$ and $M_n^{++}(\C),$ respectively. The spectrum of a matrix $A \in M_n(\C)$ is denoted by $\sigma(A).$ We denote the interior of a set $X$ by $\mathrm{int}(X).$ Throughout this paper, the symbol $\cH$ stands for a Hilbert space (either real or complex), and $\bh$ denotes the set of all bounded linear operators on $\cH.$ We denote the sets of all self-adjoint, positive semidefinite, and positive definite bounded linear operators on $\cH$ by $\bh^{sa}, \bh^{+},$ and $\bh^{++},$ respectively. There is a natural partial order on the set of the self-adjoint operators, the \emph{L\"owner order.} This is the partial order defined by the cone of the positive semidefinite operators, that is, $A \leq B$ if and only if $B-A \in \bh^{+}.$ Recall that the set $M_n^{sa}(\C)$ is a real Hilbert space of dimension $n^2$ with the \emph{Hilbert-Schmidt} inner product $\inner{A}{B}_{HS}=\tr A B.$
%In the case of complex Hilbert space we follow the convention that the inner product is linear in the first variable and conjugate-linear in the second variable.
\par
If $f: \R \supseteq \mathrm{dom} f \rightarrow \R$ is a function and $A \in M_n^{sa}(\C)$ such that $\sigma(A) \subset \mathrm{dom} f$ then the expression $f(A)$ is defined by the functional calculus, that is, $f(A)=\sum_k f\ler{\lambda_k} P_k,$ where the $\lambda_k'$s are the eigenvalues and the $P_k$'s are the eigenprojections of $A=\sum_k \lambda_k P_k.$ If $\mathrm{dom} f$ is an open interval $(a,b) \subseteq \R$ and $f \in C^1(a,b)$ then the map
\be \label{eq:func-calc}
\left\{A \in M_n^{sa}(\C) \, \middle| \, \sigma(A) \subset (a,b)\right\} \rightarrow M_n^{sa}(\C); \, A \mapsto f(A)
\ee
is \emph{Fr\'echet differentiable} at every point of its domain. The Fr\'echet derivative of the map \eqref{eq:func-calc} at the point $A$ is denoted by $\mathbf{D}f [A]\{.\}$ (or simply $\mathbf{D}f[A]$), and $\mathbf{D}f [A]\{.\} \in \cB\ler{M_n^{sa}(\C)}^{sa}.$ For every $B \in M_n^{sa}(\C),$ the equation
$$
\lim_{t \to 0} \frac{1}{t}\ler{f\ler{A+tB}-f(A)}=\mathbf{D}f[A]\{B\}
$$
holds. Moreover, $\mathbf{D}f [A]\{.\}$ can be given in a very concrete form. Assume that $A \in M_n^{sa}(\C)$ admits the spectral decomposition $A=\sum_{k=1}^n \lambda_k v_k \otimes v_k,$ where $v_1, \dots, v_n$ is an orthonormal basis in $\C^n$ and for $u,v \in \C^n$ the symbol $u \otimes v$ denotes the linear map $u \otimes v: \C^n \rightarrow \C^n; \, w \mapsto  u \otimes v (w):=\inner{w}{v} u.$
Let us introduce the short notation $E_{jk}:=v_j \otimes v_k \, (j,k \in \{1, \dots, n\}).$ The set of "matrix units" $\left\{E_{jk}\right\}_{j,k=1}^n$ form an orthonormal basis of $M_n^{sa}(\C)$ with respect to the Hilbert-Schmidt inner product.
The map $\mathbf{D}f [A]\{.\} \in \cB\ler{M_n^{sa}(\C)}^{sa}$ is given by
\be \label{eq:frech-diff}
M_n^{sa}(\C) \ni \sum_{j,k=1}^n \alpha_{jk} E_{jk} \mapsto \sum_{j,k=1}^n \alpha_{jk} f^{[1]}\left[\l_j, \l_k\right] E_{jk}, \text{ where } f^{[1]}\left[x,y\right]=\begin{cases}\frac{f(x)-f(y)}{x-y} \text{ if } x \neq y, \\ f'(x) \text{ if } x=y.\end{cases}
\ee
For details, the reader is advice to consult \cite[Thm. 3.33]{HP-book} and also \cite[Thm. 3.25]{HP-book}.

\begin{definition} \label{def:jensen-div}
Let $I \subseteq \R$ be an interval and let $f$ be a continuous convex function defined on $I.$ Set $\l \in (0,1).$ The \emph{quantum Jensen $(f,\l)$-divergence} of the self-adjoint matrices $A$ and $B$ with $\sigma(A) \cup \sigma(B) \subseteq I$ is denoted by $J_{f,\l}(A,B)$ and it is defined by
\be \label{eq:q-jensen-div}
J_{f, \l}(A,B):=(1-\l)\tr f(A)+\l \tr f(B) -\tr f \ler{(1-\l)A+\l B}.
\ee
\end{definition}

\section{The main result}

In the paper \cite{chen-tropp} Chen and Tropp introduced a function class which they called the \emph{Matrix Entropy Class.} Let us recall the definition of that function class.

\begin{definition} \label{def:MEC}
The \emph{Matrix Entropy Class} consists of functions $f: [0, \infty) \rightarrow \R$ which are either affine or satisfy the following conditions:
\bit
\item $f$ is continuous and convex, and $f \in C^2((0,\infty)),$
\item for every $n \in \N$ and for every $A \in M_n^{++}(\C)$ the linear map $\mathbf{D}f'[A] \in \cB\ler{M_n^{sa}(\C)}^{sa}$ is invertible and the map $M_n^{++}(\C) \rightarrow \cB\ler{M_n^{sa}(\C)}^{sa};\ A \mapsto \ler{\mathbf{D}f'[A]}^{-1}$ is concave with respect to the L\"owner order.
\eit
\end{definition}

\begin{remark} \label{rem:of}
Note that for a convex function $f \in C^2((0,\infty))$ the existence of $\ler{\mathbf{D}f'[A]}^{-1}$ for every $A \in M_n^{++}(\C)$ is equivalent to the property $f''>0$ on $(0,\infty).$ Indeed, convexity implies that $f''\geq 0$ on $(0, \infty),$ and if $f''(c)=0$ for some $c \in (0,\infty)$ then formula \eqref{eq:frech-diff} clearly shows that $\mathbf{D}f'[c I_n]=0 \in \cB\ler{M_n^{sa}(\C)}^{sa}$ which is not invertible. (Here the symbol $I_n$ denoted the identity element of $M_n(\C).$) On the other hand, formula \eqref{eq:frech-diff} shows also that the property $f''>0$ ensures that $\mathbf{D}f'[A]$ is a positive definite and hence invertible operator for every $A \in M_n^{++}(\C).$
\end{remark}

\begin{remark} \label{rem:kernel}
Let us introduce the notation $\mc{C}:=\left\{f \,\middle|\, f: [0, \infty) \rightarrow \R, \, f \text{ is continuous and convex} \right\}.$ Clearly, $\mc{C}$ is a convex cone. Let $\lambda \in (0,1)$ be fixed. The map
$$
\mc{C} \rightarrow [0,\infty)^{M_n^{+}(\C) \times M_n^{+}(\C)}; \, f \mapsto J_{f, \l}(.,.)
$$
is additive and positive homogeneous, and the kernel coincides with the set of the affine functions. That is, $J_{f, \l}(.,.)=0$ if and only if $f(x)=\alpha x + \beta$ for some $\alpha, \beta \in \R.$ This means that affine functions are absolutely not interesting from the viewpoint of Jensen divergences.  
\end{remark}
Now we are in the position to state our main result that characterizes the jointly convex quantum Jensen divergences.
\begin{theorem} \label{thm:main}
Let $f$ be a continuous convex function on $[0, \infty)$ such that $f \in C^2((0, \infty))$ and $f''>0$ on $(0, \infty).$ The followings are equivalent.
\ben [label=(\Roman*)]
\item \label{prop:mec} The function $f$ belongs to the Matrix Entropy Class.
\item \label{prop:jc-jen-div} The map
$$
J_{f,\l}(.,.): \, M_n^{++}(\C) \times M_n^{++}(\C) \rightarrow [0,\infty); \, (A,B) \mapsto J_{f, \l}(A,B)
$$
is jointly convex.
%that is,
%$$
%J_{f, \l}\ler{\sum_{j=1}^m \alpha_j A_j, \sum_{j=1}^m \alpha_j B_j}\leq \sum_{j=1}^m \alpha_j J_{f, \l}\ler{A_j, B_j}
%$$
%for any $A_1, \dots, A_m, B_1, \dots, B_m \in M_n^{++}$ and $\alpha_1, \dots, \alpha_m \in [0,1]$ with $\sum_{j=1}^m \alpha_j=1.$ 
\een
\end{theorem}

\section{The proof of the main result}

The proof of Theorem \ref{thm:main} has essentially two ingredients. The first one is described in Subsection \ref{subsec:int-rep} while Subsection \ref{subsec:conc-conv} is devoted to the second one. Subsection \ref{subsec:proof} contains the proof itself which refers to the results of the previous subsections.

\subsection{An integral representation of Jensen divergences} \label{subsec:int-rep}
The following integral representation of quantum Jensen divergences turns out to be useful in the sequel, although it looks a bit more complicated than the definition itself.

\begin{claim} \label{claim:int-rep}
Assume that $I \subseteq \R$ is an interval and $f: I\rightarrow \R$ is continuous and convex. Furthermore, assume that $f \in C^2\ler{\mathrm{int}(I)}$ and $A, B \in M_n^{sa}(\C)$ with $\sigma(A) \cup \sigma(B) \subset \mathrm{int}(I).$ Then the quantum Jensen $(f, \l)$-divergence $J_{f, \l}(A,B)$ (defined in \eqref{eq:q-jensen-div}) admits the following integral representation:
\be \label{eq:jen_int_repr}
J_{f,\l}(A,B)=\l (1-\l)\int_{0}^{1}\ler{(1-t)\l+t(1-\l)} \int_{0}^1 \inner{\mathbf{D}f'\left[\xi_{\l,t,s}(A,B)  \right]\left\{B-A\right\}}{B-A}_{HS} \ds \dt,
\ee
where
\be \label{eq:xi-def}
\xi_{\l,t,s}(A,B)=A+t \l (B-A)+s\ler{(1-t)\l+t(1-\l)}(B-A).
\ee
\end{claim}

\begin{proof}
Clearly,
$$
J_{f, \l}(A,B)=\l \ler{\tr f(B)-\tr f\ler{(1-\l)A+\l B}}+(1-\l)\ler{\tr f(A)-\tr f\ler{(1-\l)A+\l B}}.
$$
%$$
%=\l \left[\tr f\ler{A+\l(B-A)+t(1-\l)(B-A)}\right]_{t=0}^1 -(1-\l)\left[\tr f\ler{A+t \l (B-A)}\right]_{t=0}^1
%$$
By the Newton-Leibniz formula, the above equation can be written as follows:
$$
J_{f, \l}(A,B)=\l \int_0^1 \frac{\mathrm{d}}{\dt}\left\{\tr f\ler{A+\l(B-A)+t(1-\l)(B-A)}\right\} \dt -(1-\l) \int_0^1 \frac{\mathrm{d}}{\dt} \left\{\tr f\ler{A+t \l (B-A)}\right\} \dt
$$
Now we use the well-known identity (see, e.g., \cite[Thm. 3.23]{HP-book})
$$
\frac{\mathrm{d}}{\mathrm{d}\tau} \tr f\ler{X+\tau Y}_{|\tau=0}=\tr f'(X)Y
$$
to deduce that
$$
J_{f, \l}(A,B)=\l \int_0^1 \tr f'\ler{A+\l(B-A)+t(1-\l)(B-A)} (1-\l)(B-A) \dt -
$$
$$
-(1-\l) \int_0^1 \tr f'\ler{A+t \l (B-A)}\l (B-A) \dt,
$$
that is,
\be \label{eq:jlf1}
J_{f, \l}(A,B)=\l (1-\l)\int_0^1\tr \left\{f'\ler{A+\l(B-A)+t(1-\l)(B-A)}-f'\ler{A+t \l (B-A)}\right\}(B-A) \dt.
\ee
The (appropriate version of the) Newton-Leibniz formula tells us that
$$
f'\ler{A+\l(B-A)+t(1-\l)(B-A)}-f'\ler{A+t \l (B-A)}
$$
$$
=\int_0^1 \frac{\mathrm{d}}{\ds} \left\{f'\ler{A+t\l (B-A)+s\ler{(1-t)\l +t(1-\l)}(B-A)}\right\}\ds.
$$
Clearly, this latter expression is equal to
\be \label{eq:belso-int-rep}
\int_0^1 \mathbf{D}f' \left[A+t\l (B-A)+s\ler{(1-t)\l +t(1-\l)}(B-A)\right]\left\{\ler{(1-t)\l +t(1-\l)}(B-A)\right\}\ds.
\ee
Now let us use the notation $\xi_{\l,t,s}(A,B)=A+t \l (B-A)+s\ler{(1-t)\l+t(1-\l)}(B-A)$ and substitute the expression \eqref{eq:belso-int-rep} into \eqref{eq:jlf1}. We get that
$$
J_{f, \l}(A,B)=\l (1-\l)\int_0^1\tr \left\{\int_0^1 \mathbf{D}f' \left[\xi_{\l,t,s}(A,B)\right]\left\{\ler{(1-t)\l +t(1-\l)}(B-A)\right\}\ds\right\}(B-A) \dt
$$
$$
=\l (1-\l)\int_0^1 \ler{(1-t)\l +t(1-\l)}  \int_0^1 \tr \mathbf{D}f' \left[\xi_{\l,t,s}(A,B)\right]\left\{B-A\right\}(B-A)\ds \dt
$$
$$
=\l (1-\l)\int_{0}^{1}\ler{(1-t)\l+t(1-\l)} \int_{0}^1 \inner{\mathbf{D}f'\left[\xi_{\l,t,s}(A,B)  \right]\left\{B-A\right\}}{B-A}_{HS} \ds \dt,
$$
the proof is done.
\end{proof}

\subsection{The equivalence of a certain concavity property and a certain convexity property} \label{subsec:conc-conv}
The following claim is an essential part of our argument. This statement appeared in the proof of Thm. 1 in \cite{pv15} in a somewhat hidden way, and it also appeared in \cite[Thm 2.1]{hansen-zhang} in a slightly more concrete form and with a different proof.

\begin{claim} \label{claim:conc-conv}
Let $\ler{\cH, \inner{.}{.}}$ be a Hilbert space, let $C$ be a convex set and let $\varphi: C \rightarrow \bh^{++}$ be a map. The followings are equivalent.
\ben[label=(\Alph*)]
\item \label{prop:conc} The map $\psi: C \rightarrow \bh^{++}; \, x \mapsto \psi(x):=\ler{\varphi(x)}^{-1}$ is concave with respect to the L\"owner order.
\item \label{prop:conv} The map $\eta: C \times \cH \rightarrow [0,\infty); \, (x,v) \mapsto\eta(x,v):=\inner{\varphi(x)v}{v}$ is convex.
\een
\end{claim}

\begin{proof}
Let us show the direction \ref{prop:conc} $\Longrightarrow$ \ref{prop:conv} first. Let $x_1, \dots, x_m \in C$ and $v_1, \dots, v_m \in \cH$ be arbitrary and set $\alpha_1,\dots, \alpha_m \in [0,1]$ such that $\sum_{j=1}^m \alpha_j=1.$
By the concavity assumption \ref{prop:conc} we have
\be \label{eq:1conc}
\ler{\varphi\ler{\sum_{j=1}^m \alpha_j x_j}}^{-1} \geq \sum_{j=1}^m \alpha_j \ler{\varphi\ler{x_j}}^{-1}.
\ee
It is well-known that the function $t \mapsto 1/t$ is operator monotone decreasing on $(0,\infty),$ hence \eqref{eq:1conc} is equivalent to
\be \label{eq:2conc}
\varphi\ler{\sum_{j=1}^m \alpha_j x_j} \leq \ler{\sum_{j=1}^m \alpha_j \ler{\varphi\ler{x_j}}^{-1}}^{-1}.
\ee
By the definition of the L\"owner order, \eqref{eq:2conc} implies that
\be \label{eq:ingr1}
\inner{\varphi\ler{\sum_{j=1}^m \alpha_j x_j} \ler{\sum_{j=1}^m \alpha_j v_j}}{\sum_{j=1}^m \alpha_j v_j}
\leq \inner{\ler{\sum_{j=1}^m \alpha_j \ler{\varphi\ler{x_j}}^{-1}}^{-1}\ler{\sum_{j=1}^m \alpha_j v_j}}{\sum_{j=1}^m \alpha_j v_j}.
\ee
It is well-known that for any Hibert space $\cH$ the map
\be \label{eq:lieb-conv}
\zeta: \bh^{++} \times \cH \rightarrow [0,\infty);\, (T,v) \mapsto \zeta(T,v):=\inner{T^{-1} v}{v}
\ee
is convex. (This statement can be obtained easily as a consequence of \cite[Thm. 1]{lieb-rusk-schwarz}, and it appears in this form in, e.g., \cite[Prop. 4.3]{hansen-j-stat}.) By \eqref{eq:lieb-conv}, we obtain that
\be \label{eq:ingr2}
\inner{\ler{\sum_{j=1}^m \alpha_j \ler{\varphi\ler{x_j}}^{-1}}^{-1}\ler{\sum_{j=1}^m \alpha_j v_j}}{\sum_{j=1}^m \alpha_j v_j}
\leq
\sum_{j=1}^m \alpha_j \inner{\ler{\ler{\varphi\ler{x_j}}^{-1}}^{-1}v_j}{v_j}.
\ee
Inequalities \eqref{eq:ingr1} and \eqref{eq:ingr2} together imply the desired convexity property
$$
\inner{\varphi\ler{\sum_{j=1}^m \alpha_j x_j} \ler{\sum_{j=1}^m \alpha_j v_j}}{\sum_{j=1}^m \alpha_j v_j}
\leq
\sum_{j=1}^m \alpha_j \inner{\varphi\ler{x_j} v_j}{v_j}.
$$

To prove the direction \ref{prop:conv} $\Longrightarrow$ \ref{prop:conc} let us consider an arbitrary vector $u \in \cH.$ Let us define
\be \label{eq:wi-def}
w_i:=\ler{\varphi\ler{x_i}}^{-1} \circ \ler{\sum_{j=1}^m \alpha_j \ler{\varphi\ler{x_j}}^{-1}}^{-1} (u).
\ee
Observe that $\sum_{i=1}^m \alpha_i w_i= u.$
Therefore, the expression $\inner{\varphi\ler{\sum_{i=1}^m \alpha_i x_i} u}{u}$ can be written in the form $\inner{\varphi\ler{\sum_{i=1}^m \alpha_i x_i} \ler{\sum_{i=1}^m \alpha_i w_i}}{\sum_{i=1}^m \alpha_i w_i}$ and $\inner{\ler{\sum_{j=1}^m \alpha_j \ler{\varphi\ler{x_j}}^{-1}}^{-1} (u)}{u}$ can be written as $\sum_{i=i}^m \alpha_i \inner{\varphi\ler{x_i} w_i}{w_i}$ because
$$
\sum_{i=i}^m \alpha_i \inner{\varphi\ler{x_i} w_i}{w_i}=
\sum_{i=i}^m \alpha_i \inner{\varphi\ler{x_i} \circ \ler{\varphi\ler{x_i}}^{-1} \circ \ler{\sum_{j=1}^m \alpha_j \ler{\varphi\ler{x_j}}^{-1}}^{-1} (u)}{w_i}
$$
$$
=
\inner{\ler{\sum_{j=1}^m \alpha_j \ler{\varphi\ler{x_j}}^{-1}}^{-1} (u)}{\sum_{i=1}^m \alpha_i w_i}
=
\inner{\ler{\sum_{j=1}^m \alpha_j \ler{\varphi\ler{x_j}}^{-1}}^{-1} (u)}{u}.
$$
By the convexity assumption \ref{prop:conv} we have
$$
\inner{\varphi\ler{\sum_{i=1}^m \alpha_i x_i} \ler{\sum_{i=1}^m \alpha_i w_i}}{\sum_{i=1}^m \alpha_i w_i}
\leq
\sum_{i=i}^m \alpha_i \inner{\varphi\ler{x_i} w_i}{w_i}
$$
which is equivalent to
$$
\inner{\varphi\ler{\sum_{i=1}^m \alpha_i x_i} u}{u} \leq \inner{\ler{\sum_{j=1}^m \alpha_j \ler{\varphi\ler{x_j}}^{-1}}^{-1} (u)}{u}
$$
by the above observations. The vector $u$ was arbitrary, which means that we have
\be \label{eq:majd-conc}
\varphi\ler{\sum_{i=1}^m \alpha_i x_i} \leq \ler{\sum_{j=1}^m \alpha_j \ler{\varphi\ler{x_j}}^{-1}}^{-1}
\ee
in the L\"owner sense. Taking the inverse reverses the order on positive definite operators, hence \eqref{eq:majd-conc} is equivalent to the desired concavity property
$$
\ler{\varphi\ler{\sum_{i=1}^m \alpha_i x_i}}^{-1} \geq \sum_{j=1}^m \alpha_j \ler{\varphi\ler{x_j}}^{-1}.
$$
\end{proof}

\subsection{The proof of Theorem \ref{thm:main}} \label{subsec:proof}
\begin{proof}
Let us consider the direction \ref{prop:mec} $\Longrightarrow$ \ref{prop:jc-jen-div}. Assume that the map $M_n^{++}(\C) \rightarrow \cB\ler{M_n^{sa}(\C)}^{sa};\ A \mapsto \ler{\mathbf{D}f'[A]}^{-1}$ is concave with respect to the L\"owner order. Note that by the property $f''>0$ the operator $\mathbf{D}f'[A]$ is positive definite for every $A \in M_n^{++}(\C)$ (see Remark \ref{rem:of}). This means that the image of the map $A \mapsto \mathbf{D}f'[A]$ is contained in $\cB\ler{M_n^{sa}(\C)}^{++}.$ Therefore, we can apply Claim \ref{claim:conc-conv} and deduce that the map
\be \label{eq:konkr-konv}
M_n^{++}(\C) \times M_n^{sa}(\C) \rightarrow [0, \infty); \, (X,Y) \mapsto \inner{\mathbf{D}f'[X]\{Y\}}{Y}_{HS}
\ee
is jointly convex. Furthermore, for fixed $\l \in (0,1)$ and $t,s \in [0,1]$ the map $\xi_{\l,t,s}: \, (A, B) \mapsto \xi_{\l,t,s}(A,B)$ is affine, that is, $\xi_{\l,t,s}\ler{\sum_{j=1}^m \alpha_j A_j, \sum_{j=1}^m \alpha_j B_j}=\sum_{j=1}^m \alpha_j \xi_{\l,t,s}\ler{A_j,B_j}.$ (For the definition of $\xi_{\l,t,s}$ see formula \eqref{eq:xi-def}.) The convexity of \eqref{eq:konkr-konv} and the affine property of $\xi_{\l,t,s}$ together imply that
$$
\inner{\mathbf{D}f'\left[\xi_{\l,t,s}\ler{\sum_{j=1}^m \alpha_j A_j, \sum_{j=1}^m \alpha_j B_j}\right]\left\{\sum_{j=1}^m \alpha_j B_j-\sum_{j=1}^m \alpha_j A_j\right\}}{\sum_{j=1}^m \alpha_j B_j-\sum_{j=1}^m \alpha_j A_j}_{HS}
$$
$$
=\inner{\mathbf{D}f'\left[\sum_{j=1}^m \alpha_j \xi_{\l,t,s}\ler{A_j,B_j}\right]\left\{\sum_{j=1}^m \alpha_j \ler{B_j-A_j}\right\}}{\sum_{j=1}^m \alpha_j \ler{B_j-A_j}}_{HS}
$$
$$
\leq \sum_{j=1}^m \alpha_j \inner{\mathbf{D}f'\left[\xi_{\l,t,s}\ler{A_j,B_j}\right]\left\{B_j-A_j\right\}}{B_j-A_j}_{HS}
$$
holds for all $A_1, \dots A_m, B_1, \dots, B_m \in M_n^{++}(\C)$ and $\alpha_1, \dots, \alpha_m \in [0,1]$ with $\sum_{j=1}^m \alpha_j =1.$ Consequently,
\footnotesize
$$
\l (1-\l)\int_{0}^{1}\ler{(1-t)\l+t(1-\l)} \int_{0}^1
\inner{\mathbf{D}f'\left[\xi_{\l,t,s}\ler{\sum_{j=1}^m \alpha_j A_j, \sum_{j=1}^m \alpha_j B_j}\right]\left\{\sum_{j=1}^m \alpha_j B_j-\sum_{j=1}^m \alpha_j A_j\right\}}{\sum_{j=1}^m \alpha_j B_j-\sum_{j=1}^m \alpha_j A_j}_{HS}\ds \dt
$$
\normalsize
$$
\leq 
\l (1-\l)\int_{0}^{1}\ler{(1-t)\l+t(1-\l)} \int_{0}^1
\ler{\sum_{j=1}^m \alpha_j \inner{\mathbf{D}f'\left[\xi_{\l,t,s}\ler{A_j,B_j}\right]\left\{B_j-A_j\right\}}{B_j-A_j}_{HS} }\ds \dt.
$$
By the result of Claim \ref{claim:int-rep}, the above inequality is nothing else but $J_{f,\l}\ler{\sum_{j=1}^m \alpha_j A_j, \sum_{j=1}^m \alpha_j B_j} \leq \sum_{j=1}^m \alpha_j J_{f,\l} \ler{A_j, B_j},$ which is the desired joint convexity property.
\par
Now we prove the direction \ref{prop:jc-jen-div} $\Longrightarrow$ \ref{prop:mec}.
Let $A_1, \dots A_m \in M_n^{++}(\C), \, B_1, \dots, B_m \in M_n^{sa}(\C)$ and $\alpha_1, \dots, \alpha_m \in [0,1]$ with $\sum_{j=1}^m \alpha_j =1$ be arbitrary but fixed. Let us choose $\eps_0>0$ such that $A_j+\eps B_j \in M_n^{++}(\C)$ for every $0 < \eps <\eps_0$ and $j \in \{1,\dots, m\}.$ By assumption, we have 
$$
J_{f,\l}\ler{\sum_{j=1}^m \alpha_j A_j, \sum_{j=1}^m \alpha_j A_j + \eps \sum_{j=1}^m \alpha_j B_j }\leq \sum_{j=1}^m \alpha_j J_{f,\l}\ler{A_j, A_j+\eps B_j}.
$$
By the integral representation of the Jensen divergences (Claim \ref{claim:int-rep}), the above inequality is equivalent to
\small
$$
\eps^2 \l (1-\l)\int_{0}^{1}\ler{(1-t)\l+t(1-\l)} \int_{0}^1 \inner{\mathbf{D}f'\left[\xi_{\l,t,s}\ler{\sum_{j=1}^m \alpha_j A_j, \sum_{j=1}^m \alpha_j A_j + \eps \sum_{j=1}^m \alpha_j B_j }  \right]\left\{\sum_{j=1}^m \alpha_j B_j\right\}}{\sum_{j=1}^m \alpha_j B_j}_{HS} \ds \dt
$$
\normalsize
\be \label{eq:egyen}
\leq \eps^2 \sum_{j=1}^m \alpha_j \l (1-\l)\int_{0}^{1}\ler{(1-t)\l+t(1-\l)} \int_{0}^1 \inner{\mathbf{D}f'\left[\xi_{\l,t,s}\ler{A_j,A_j+\eps B_j}  \right]\left\{B_j\right\}}{B_j}_{HS} \ds \dt.
\ee
The continuity of $f''$ ensures that the map $X \mapsto \mathbf{D} f' [X]$ is continuous. It is clear by the formula \eqref{eq:xi-def} that $\lim_{\eps \to 0} \xi_{\l,t,s}\ler{\sum_{j=1}^m \alpha_j A_j, \sum_{j=1}^m \alpha_j A_j + \eps \sum_{j=1}^m \alpha_j B_j }=\sum_{j=1}^m \alpha_j A_j$ and $\lim_{\eps \to 0} \xi_{\l,t,s}\ler{A_j,A_j+\eps B_j}=A_j.$ Therefore, if we divide the inequality \eqref{eq:egyen} by $\eps^2$ and take the limit $\eps \to 0$ then we get that
$$
\l (1-\l)\int_{0}^{1}\ler{(1-t)\l+t(1-\l)} \int_{0}^1 \inner{\mathbf{D}f'\left[\sum_{j=1}^m \alpha_j A_j \right]\left\{\sum_{j=1}^m \alpha_j B_j\right\}}{\sum_{j=1}^m \alpha_j B_j}_{HS} \ds \dt
$$
$$
\leq \sum_{j=1}^m \alpha_j \l (1-\l)\int_{0}^{1}\ler{(1-t)\l+t(1-\l)} \int_{0}^1 \inner{\mathbf{D}f'\left[A_j  \right]\left\{B_j\right\}}{B_j}_{HS} \ds \dt,
$$
that is,
$$
\inner{\mathbf{D}f'\left[\sum_{j=1}^m \alpha_j A_j \right]\left\{\sum_{j=1}^m \alpha_j B_j\right\}}{\sum_{j=1}^m \alpha_j B_j}_{HS}
\leq \sum_{j=1}^m \alpha_j \inner{\mathbf{D}f'\left[A_j  \right]\left\{B_j\right\}}{B_j}_{HS}.
$$
So we deduced that the map given by \eqref{eq:konkr-konv} is jointly convex. By Claim \ref{claim:conc-conv} this means that the map $A \mapsto \ler{\mathbf{D}f'[A]}^{-1}$ is concave with respect to the L\"owner order, and hence the function $f$ belongs to the Matrix Entropy Class. The proof is done.
\end{proof}

\bibliographystyle{amsplain}

\end{document}